\documentclass[lettersize,journal]{IEEEtran}
\usepackage{amsmath,amsfonts}
\usepackage{algorithmic}
\usepackage{algorithm}
\usepackage{array}
\usepackage{textcomp}
\usepackage{stfloats}
\usepackage{url}
\usepackage{verbatim}
\usepackage{graphicx}
\usepackage{cite}
\usepackage{physics}
\usepackage{bm}
\usepackage[short,c2]{optidef}
\usepackage{cite}
\usepackage{amsmath,amssymb,amsfonts, amsthm}
\usepackage{algorithmic}
\usepackage{graphicx}
\usepackage{textcomp}
\usepackage{tabularx}
\usepackage{epstopdf}
\usepackage{graphicx}
\usepackage{colortbl}
\usepackage{blkarray}
\usepackage{amsmath}
\usepackage{arydshln}
\usepackage{mathtools}
\usepackage{nicematrix} 
\usepackage{physics}
\usepackage{algorithm}
\newtheorem{theorem}{Theorem}
\newtheorem*{remark}{Remark}

\newtheorem{corollary}{Corollary}

\usepackage{tikz}
\usetikzlibrary{shapes,shadows,automata,arrows,positioning,calc}
\usetikzlibrary{backgrounds}

\begin{document}

\title{Query-Based Sampling of Heterogeneous CTMCs: Modeling and Optimization with Binary Freshness}
\author{Nail~Akar,~\IEEEmembership{Senior~Member,~IEEE,}~and~Sennur~Ulukus,~\IEEEmembership{Fellow,~IEEE}
\thanks{N.~Akar is with the Electrical and Electronics Engineering Department, Bilkent University, Ankara, Turkey 
(e-mail:akar@ee.bilkent.edu.tr) 

This work is done when N.~Akar is on sabbatical leave at University of Maryland, MD, USA. The work of N.~Akar is supported in part by the Scientific and Technological Research Council of Turkey 
(T\"{u}bitak) 2219-International Postdoctoral Research Fellowship Program.} 
\thanks{S.~Ulukus is with the Department of Electrical and Computer Engineering, University of Maryland, MD, USA (e-mail:ulukus@umd.edu).}}

\maketitle

\begin{abstract}
We study a remote monitoring system in which a mutually independent and heterogeneous collection of finite-state irreducible continuous time Markov chain (CTMC) based information sources is considered. In this system, a common remote monitor queries the instantaneous states of the individual CTMCs according to a Poisson process with possibly different intensities across the sources, in order to maintain accurate estimates of the original sources. \color{black}Three information freshness models are considered to quantify the accuracy of the remote estimates: fresh when equal (FWE), fresh when sampled (FWS) and fresh when close (FWC). For each of these freshness models, closed-form expressions are derived for mean information freshness for a given source. Using these expressions, optimum sampling rates for all sources are obtained so as to maximize the weighted sum freshness of the monitoring system, subject to an overall sampling rate constraint. This optimization problem leads to a water-filling solution with quadratic worst case computational complexity in the number of information sources. Numerical examples are provided to validate the effectiveness of the optimum sampling policy in comparison to several baseline sampling  policies. 

\end{abstract}


\section{Introduction} \label{section:Introduction}
Timely delivery of status update packets from a number of information sources for maintaining information freshness at a remote monitor (or destination) has recently gained significant attention for the development of applications requiring real-time monitoring and control.
In such applications, information sources generate status update packets that contain the samples of an underlying random process, e.g., a sensor sampling a physical quantity such as temperature, humidity, etc., or an information item with time-varying content such as news, weather reports, etc., which are subsequently delivered to the monitor through a communication network. In the push-based communication paradigm, information sources decide when to sample and form the information packets, which are subsequently forwarded towards the destination \cite{Yates__SBR, Moltafet__MultiSourceQueue}. On the other hand, in query-based (or pull-based) communications,  monitors proactively query the information sources upon which sampling takes place \cite{pullBased_1,pullBased_2}. 

For the design and optimization of status update systems, a need for quantifying information freshness is evident. 
For this purpose, age of information (AoI) is commonly used which is a continuous-valued continuous-time stochastic process maintained at the destination that keeps track of the time elapsed since the reception of the last status update packet received from a particular information source. 
The AoI process was first introduced in \cite{Yates__HowOftenShouldOne} for a single-source M/M/1 queuing model which resulted in substantial interest in the modeling of AoI and its optimization in very general settings including multiple sources \cite{kosta_etal_survey, RoyYates__AgeOfInfo_Survey}. 

In this paper, the interest will be on the remote estimation of discrete-valued information sources, in particular CTMC based information sources, for which the status of the sources change at random instants, and the source dynamics (set of states and transition rates) is known at the monitor. In this setting, the monitor and the source are synchronized (in sync or fresh) when the monitor's estimate overlaps with the state of the source. Otherwise, they are de-synchronized (out of sync or erroneous or stale). Although performance metrics derived from AoI have played an important role for the development of status update systems for continuous-valued information sources, these metrics are not as suitable for discrete-valued information sources with known dynamics, the latter being the focus of this paper. For example, even when the monitor has perfect
knowledge of the source when its discrete status stays intact, its corresponding AoI process would continue to increase
with time, leading to an undesired penalty \cite{maatouk_etal_ton20}. 

On the basis of such shortcomings of AoI, alternative performance metrics have been proposed including age of synchronization (AoS) which is defined as the elapsed time since the content at the destination has de-synchronized with the source \cite{AoS_1__Yates__TwoFreshnessMetricsForLocal, AoS_2__SchedulingToMinimizeAoSInWireless, AoS_3__TimelySynchronizedSporadic}. With AoS, the penalty of the monitor having an erroneous estimate of the source increases linearly with time as long as the de-synchronization condition stays. A more general process proposed recently is age of incorrect information (AoII) which is defined as the product of an increasing time penalty function and 
another penalty function chosen general enough to depend on both the current estimate at the monitor and the actual
state of the process \cite{maatouk_etal_ton20}. AoII covers AoS as a sub-case whereas other application-oriented choices for the two penalty functions are further elaborated in \cite{maatouk_etal_ton20}.  
For certain applications including caching systems, dissatisfaction of being in an erroneous state does not necessarily increase with time but instead, with the number of versions the monitor lags with respect to the original source. For this purpose, version age (VA) has been introduced which keeps track of the number of status changes that have occurred at the source since the last time the content at the destination is de-synchronized with the source \cite{Abolhassani__VA_FreshCachingForDynamic, Yates__VA_AgeOfGossipInNetworks, melih2020infocom,
Sennur__VA_AgeOfGossipCommunityStructure}. 
In some applications, in contrast to the age-driven metrics, the information at the destination may not possess a value unless the content at the destination is in sync with the source irrespective of the duration of the erroneous state. In such applications, there is not a need for the dissatisfaction to grow with time or with outdated versions, which gives rise to the so-called {\em binary freshness} (BF) process which takes the value of one when the information at the destination is in sync with the source, and is zero otherwise. 
This simple BF metric was extensively used in the optimization of Web crawling systems by which local copies of remote Web pages are managed for Web search engines \cite{crawling1,crawling2}. BF has recently been used as a performance metric for freshness quantification in other computer and communication systems, such as in cache update systems \cite{Sennur__BF_MaximizingInfoFreshness, Sennur__BF_FreshnessBasedCache, Sennur__BF_InfoFreshnessInCacheUpdating,gamgam_akar_iot23}, gossiping networks \cite{Sennur__BF_GossipingWithBinaryFreshness,bastopcu_etal_entrop24}, infection tracking systems \cite{bastopcu_ulukus_entropy22}. The focus of this study is applications for which the use of the BF metric is suitable for freshness quantification.  
\color{black}
\begin{figure}[t]
	\centering
	\begin{tikzpicture}[scale=0.25]
		\tikzstyle{note} = [rectangle, dashed, draw, fill=white, font=\footnotesize,
		text width=20em, text centered, rounded corners, minimum height=18em]
		\draw[very thick, red](4,2) circle (2.2);
		\draw[very thick, green](4,-7) circle (2.2);
		\draw (4,0.7) node[anchor=south] {\scriptsize{$X_2(t)$}} ;
		\draw[very thick, blue] (4,8) circle (2.2)  ;
		\draw (4,6.7) node[anchor=south] {\scriptsize{$X_1(t)$}} ;
		\draw (4,-8.3) node[anchor=south] {\scriptsize{$X_N(t)$}} ;
	  \draw[ultra thick, red,<-] (7,2) -- (12,2) ;
        \draw (9.5,2) node[anchor=south] {\scriptsize{$\lambda_2$}};	
        \draw[ultra thick, blue,<-] (7,7) -- (12,5) ;
        \draw (9.9,6.3) node[anchor=south] {\scriptsize{$\lambda_1$}};	
        \draw (9.5,-3.1) node[anchor=south] {\scriptsize{$\lambda_N$}};	
	  \draw[ultra thick, green,<-] (7,-5.5) -- (12,-2) ;
	  \draw (4,-3.8) node[anchor=south] {{\Large $\vdots$}};
		\draw (9.5,-1.3) node[anchor=south] {{\Large $\vdots$}};
		\draw[thick, gray](17,2.5) circle (4);
		\filldraw (17,3.5)  node[anchor=center] {\scriptsize{$\tilde{X}_n(t)$}};	
		\filldraw (17,1.5)  node[anchor=center] {\scriptsize{$1 \leq n \leq N$}};	
		\filldraw (17.7,-2.5) node[anchor=center] {\small{remote monitor}};	
	\end{tikzpicture}
	\caption{Query-based status update system in which the monitor queries the CTMC $X_n(t)$ associated with source-$n$ with intensity $\lambda_n$, to maintain an estimate $\tilde{X}_n(t)$ of $X_n(t)$.}
	\label{fig:systemmodel}
 \vspace*{-0.5cm}
\end{figure}
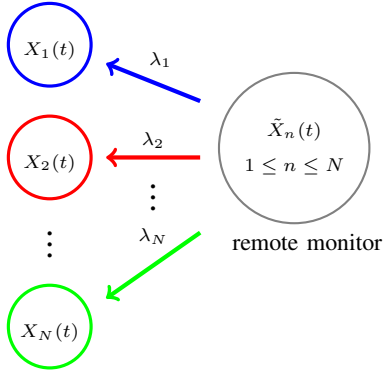	 
In this paper, we consider the query-based communication system in Fig.~\ref{fig:systemmodel} involving a number of CTMC-based information sources each of which is queried by a common remote monitor according to a Poisson process after which samples are taken and sent as status update packets towards the monitor, to maintain a remote estimate of the original process. Queries and status update packet transmissions are assumed to take place immediately. 
The individual CTMCs are finite-state and irreducible, mutually independent, and their heterogeneous dynamics (state-space and transition rates do not have to be the same across the sources) are known at the monitor. We employ a martingale estimator that is also used in similar studies such as \cite{maatouk_etal_ton20,kam_etal_infocom20,bastopcu_ulukus_entropy22}, for which the remote estimate of the source between two consecutive samples, is the value of the first sample. 
The martingale estimator is easy to maintain, and moreover, we will show that it is possible to obtain closed-form expressions for several freshness metrics with this estimator, which subsequently leads to an algorithmic optimization procedure for obtaining the optimum sampling rates in a heterogeneous setting. 
Non-Poisson query arrivals, such as periodic queries, more advanced estimators such as the maximum likelihood estimator (MLE), and non-zero transmission times for query and status update packets, are left for future research. 

We employ three non-age driven freshness models, two being binary, to quantify the accuracy of the martingale estimators. In fresh when equal (FWE), the information is fresh at the destination when the remote estimator is in sync with the state of the original source. In the fresh when sampled (FWS) freshness model (which is also binary), the information becomes stale when the state of the original source changes and it becomes fresh only when a new sample is taken.  
For FWS, when the freshness process is stale, it cannot be set until the remote monitor re-samples the process, which is in contrast with FWE. The advantage of the FWS model is that it represents certain applications where freshness can be regained only with an explicit user action once it has been lost, and further, closed form expressions for mean freshness can be obtained for FWS for a very general class of CTMCs.
The non-binary fresh when close (FWC) model is a generalization of the binary FWE model where the level of freshness depends on the proximity of the original process and its estimator at the monitor, yielding more flexibility than FWE.   
\color{black}

In this work, for the three freshness models of interest,  we first find closed form expressions for mean freshness as a function of the sampling rate and subsequently we find the optimum sampling rate for each source that maximizes the weighted sum freshness for the system, under a total sampling rate constraint for the monitor. 
We show that this optimization problem possesses a water-filling structure which amounts to a procedure in which a limited volume of water is poured into a pool organized into a number of bins with different ground levels with the water level in each bin giving the desired optimum solution \cite{WF_PracticalAlgorithms}. 
Water-filling based optimization has been successfully used in solving optimal resource allocation problems arising in wireless networks; see \cite{WF_Poor, WF_PracticalAlgorithms} for two related surveys, and also references \cite{exact_wf_algorithm} and \cite{WF_EnergyAllocation}, that make use of water-filling techniques. 
Iterative methods are available for water-filling based optimization which will be shown to give rise to an algorithm for the optimum monitoring problem in Fig.~\ref{fig:systemmodel} with quadratic worst-case computational complexity in the number of sources.

The contributions of our paper are as follows: 
\begin{itemize}
\item To the best of our knowledge, optimum sampling of heterogeneous CTMCs under overall sampling rate constraints has not been explored in the literature. Our work initiates this line of study. 
\item We derive closed form expressions for mean freshness for FWE, FWC and FWS models for finite-state and irreducible CTMCs. For the FWS model, the obtained expression is in terms of the sum of first order rational functions of the sampling rate which is a strictly concave increasing function. For the FWE and FWC models, similar expressions are obtained for the sub-case of time-reversible CTMCs (using the real-valued eigenvalues, and eigenvectors of the corresponding generators) that cover the well-established birth-death Markov chains that arise frequently in the performance of computer and communication systems. 
\item  The obtained expressions allow us to use computationally efficient water-filling algorithms to obtain optimum sampling policies.
\end{itemize}

\color{black}
The remainder of this paper is organized as follows. 
Related work is summarized in Section~\ref{section:Related}.
Preliminaries on CTMCs are given in Section~\ref{sec:preliminaries}.
In Section~\ref{section:SystemModel}, the detailed system model is presented. Section~\ref{section:analytical} addresses the derivation of mean freshness expressions for the three freshness models. The water-filling based optimization algorithm is presented in Section~\ref{section:Optimization}. A comparative evaluation of the proposed optimum sampler and several baseline sampling policies are presented in Section~\ref{section:Examples}. Conclusions, some open problems and potential future directions are given in Section~\ref{section:Conclusions}. 

\section{Related Work} \label{section:Related}
The majority of the remote estimation problems in the literature are in the discrete-time setting. In \cite{huang_etal_TWC20}, the authors study the remote estimation of a linear time-invariant dynamic system while focusing on the trade-off between reliability and freshness.   In \cite{yutao_ephremides_globecom21}, AoII is investigated in a status update system involving a multi-state Markovian information source, a monitor, and a channel susceptible to packet errors, and the communication goal is to minimize AoII subject to a power constraint. 
The authors of \cite{pappas_kountouris_ICAS21} study a transmitter monitoring the evolution of a two-state discrete Markov source and sending status updates to a destination over an unreliable wireless channel for the purpose of real-time source reconstruction for remote actuation. This work is then extended in \cite{fountoulakis_etal_COMLET23} with more general discrete stochastic source processes and resource constraints. The work presented in \cite{champati_etal_tcom22} studies the trade-off between the sampling frequency and staleness in detecting the events through a freshness metric called age penalty which is defined as the time elapsed since the first transition out of the most recently observed state. 
The authors of \cite{salimnejad_etal_tcom24} investigate a time-slotted communication system for tracking a discrete-time Markovian source
with joint sampling and transmission over a wireless channel.

For the continuous-time setting, \cite{sun_etal_TIT20} investigates the problem of sampling a Wiener process with samples forwarded to a remote estimator over a channel that is modeled as a queue.
The authors of \cite{inoue_takine_infocom19} investigate the effect of AoI on the accuracy of a remote monitoring system which displays the latest state information from a
CTMC and they develop a computational method for finding the conditional
probability of the displayed state, given the actual current state
of the information source.
The authors of \cite{cosandal_etal_isit24} obtain a push-based sampling policy for remote tracking of a CTMC source subject to a sampling rate constraint using constrained semi-Markov decision processes.
A common feature of the above works is the existence of a single information source which gets to be sampled. On the other hand, \cite{bastopcu_ulukus_entropy22} studies the sampling of a collection of heterogeneous two-state CTMC-based information sources each modeling whether an individual is infected with a virus or not, while using the binary freshness metric, and \cite{melih2020infocom} studies sampling of multiple heterogeneous Poisson processes representing the citation indices of multiple researchers with the goal of keeping timely estimates (similar to version age) of all the random processes. In the current manuscript, we study a heterogeneous collection of general finite-state (not necessarily two-state) CTMC-based information sources with three different freshness models (including the binary freshness metric) where the sources need to be sampled for remote estimation with a constraint on the overall sampling rate.

\section{Preliminaries}
\label{sec:preliminaries}
The focus of the current paper is on irreducible, finite-state, time-homogeneous (transition rates do not depend on time) CTMCs with their main properties given in this section based on \cite{gallager_book,norris_book}. We consider the CTMC $X(t), t \geq 0, X(t) \in \mathcal{K} = \{ 1,2,\ldots,K\}$, where $K \geq 1$ is the number of states. 
The process $X(t)$ has the infinitesimal generator $Q$ of size $K$ with its $(i,j)$th entry denoted by $q_{ij}$ which is the transition rate from state $i$ to state $j$ for $i \neq j$ and its diagonal entries are strictly negative satisfying $Q e=0$, where $e$ is a column vector of ones of appropriate size. The CTMC $X(t)$ is called irreducible if it is possible with some
positive probability to get from any state to any other state in some finite time. An irreducible CTMC does not have any transient states. Hence, the limiting probability of $X(t)$ being in state $j$ conditioned on being in state $i$ at time zero, exists and does not depend on the initial state,
\begin{align}
     \lim_{t \rightarrow \infty} \mathbb{P} [X(t)=j | X(0)=i] \label{limiting} & =: \pi_j.
\end{align}
Moreover, the row vector $\pi = \{ \pi_j \}$, known as the stationary (or invariant) distribution, satisfies 
$\pi Q =0$ and $\pi e=1$, which are known as the global balance equations (GBE).
The GBEs for CTMCs are a set of equations, one for each state $s$ of the CTMC, which states that the total probability flux out of a state $s$ should be equal to the total probability flux from other states into the state $s$, in steady-state \cite{gallager_book}. 

Irreducible, finite-state CTMCs are ergodic, i.e., for any function $f:\mathcal{K} \rightarrow \mathbb{R}$, the following holds,
\begin{align}
\lim_{\tau \rightarrow \infty} \frac{1}{\tau} \int_{t=0}^{\tau} f_{X(t)} \dd{t} &= \sum_{i=1}^K f_i \pi_i, \label{ergodic}
\end{align}
where $f_i$ is the value of the function at state $i$, known as the Ergodic Theorem \cite{norris_book}.
Note that the right hand side of \eqref{ergodic} is an expectation with respect to the invariant distribution $\pi$.
In particular, when $f_i =i$, then we have
\begin{align}
&\lim_{\tau \rightarrow \infty} \frac{1}{\tau} \int_{t=0}^{\tau} X(t) \dd{t} = \sum_{i=1}^K i \pi_i =: 
\mathbb{E} [X],
\label{ergodic2}
\end{align}
where the random variable $X$ has the same limiting distribution in \eqref{limiting}, i.e., $\mathbb{P} [X = j] =  \lim_{t \rightarrow \infty} \mathbb{P} [X(t)=j ] = \pi_j, j \in \mathcal{K}$. $X$ is called the random variable associated with the CTMC $X(t)$ in the steady-state. 

The generator $Q$ has a left eigenvector $\pi$ and right eigenvector $e$, associated with the simple eigenvalue at zero and all other eigenvalues having strictly negative real parts \cite{gallager_book}. $X(t)$ is called a time-reversible CTMC
if its generator $Q$ satisfies the following detailed balance equations (DBE),
\begin{align}
    \pi_i q_{ij} & = \pi_j q_{ji}, \quad i \neq j. \label{reversible}
\end{align}
The equations \eqref{reversible} are the GBEs for birth-death chains since the corresponding state transitions take place between neighboring states only. Therefore, birth-death chains are time-reversible.
For a time-reversible CTMC $X(t)$, let
\begin{align}
   \Pi = & {\rm diag}\{ \pi_1,\pi_2,\ldots,\pi_{K} \}, \label{Pi} 
\end{align} 
be the diagonal matrix composed of the entries of $\pi$. Also let
\begin{align} 
S & = \Pi^{1/2} Q \Pi^{-1/2}, \label{S}
\end{align} 
which is a symmetric matrix from \eqref{reversible}. Symmetric matrices have real eigenvalues and they are diagonalizable by orthogonal transformations \cite{golub_book}. Therefore, there exists an orthonormal matrix $U$ such that
\begin{align}
    U^T S U & = D = {\rm diag} \{-d_1,-d_2,\ldots,-d_{K -1},0 \},\label{diagonalization0}
\end{align}
where $-d_i$, with $d_i>0$, are the corresponding real eigenvalues of the matrix $S$. Moreover, the matrix defined by $T = \Pi^{-1/2} U$ diagonalizes the original generator $Q$,  
\begin{align}
    T^{-1} Q T & = D.  \label{diagonalization} 
\end{align}
Next, we present several properties on the left and right eigenvectors of time-reversible CTMCs.
Let the $(i,j)$th entries of $T$ and $\tilde{T}=T^{-1}$ be denoted by $t_{ij}$ and $\tilde{t}_{ij}$, respectively. The way the transformation matrix $T$ is defined, the $i$th row of ${\tilde{T}}$, i.e., the $i$th left eigenvector of $Q$, is obtained by post-multiplying by $\Pi$ the transpose of the $i$th column of $T$, i.e., the transpose of the $i$th right eigenvector of $Q$. 
Consequently, $\tilde{t}_{ji} = \pi_i t_{ij}$ for all $i,j$.
Moreover, the row vector $\pi$ is the $K$th row of $\tilde{T}$ and $e$ is the $K$th column of $T$ in \eqref{diagonalization}. 
\color{black}

\section{System Model} \label{section:SystemModel}
We consider the monitoring system in Fig.~\ref{fig:systemmodel} with $N$ continuous-time information sources each of which is a 
finite-state, irreducible CTMC. The CTMC associated with source-$n$ is denoted by $X_n(t)$, $n \in \mathcal{N} = \{1,2,\ldots,N\}$, $t \geq 0$, and $X_n(t) \in \{ 1,2,\ldots,K_n\}$, where $K_n$ is the size of state space for $X_n(t)$. The process $X_n(t)$ has the infinitesimal generator ${Q_n}$ of size $K_n$ with its $(i,j)$th entry denoted by $q_{n,ij}$. 
The steady-state vector satisfies $\bm{\pi}_n {Q_n} =0$ and $\bm{\pi}_n e=1$ with $\bm{\pi}_n = \{ \pi_{n,i} \}$.
The transition rate out of state $i$ is denoted by $\sigma_{n,i} = \sum_{j \neq i} q_{n,ij}$.
The average transition intensity of source-$n$ is denoted by $r_n$, i.e., $r_n = \sum_{i=1}^{K_n} \pi_{n,i} \sigma_{n,i}$ which is the long-term frequency of state transitions for the CTMC $X_n(t)$.
We denote by $r$ the system transition intensity, $r= \sum_{i=1}^N r_n$.

The remote monitor in Fig.~\ref{fig:systemmodel} samples the original process $X_n(t)$ according to a Poisson process with intensity $\lambda_n >0$ in order to maintain an estimate of the instantaneous state of the original information process. In this paper, we propose to use the martingale estimator $\tilde{X}_n(t)$ which is given as $\tilde{X}_n(t) = X_n(t')$, where $t'$ is the latest sampling time before $t$. 
The accuracy of the remote estimator is studied with three information freshness models described below. For the FWE information freshness model, the information is said to be fresh at the remote monitor
only when the original process and its estimate are equal, i.e., $\tilde{X}_n(t) = X_n(t)$, and is otherwise stale. In the FWE model, there is no value at all in a sample, unless the original process and its estimator are synchronized. Hence, the binary freshness process $F_{n,e}(t)$ is defined as $F_{n,e}(t)=1$ when $\tilde{X}_n(t) = X_n(t)$, and zero otherwise. It is clear that the joint process $(\tilde{X}(t),X(t))$ is also an irreducible CTMC as the original CTMC $X(t)$, and using the Ergodic Theorem \cite{norris_book}, the time-average of the freshness process $F_{n,e}(t)$ is given by, 
\begin{align}
\lim_{\tau \rightarrow \infty} \frac{1}{\tau}\int_{t=0} ^{\tau} F_{n,e}(t) \dd{t} := \mathbb{E}[F_{n,e}] &, \label{freshnessDefine}
\end{align}
where $F_{n,e}$ is the random variable associated with the random process $F_{n,e}(t)$ in steady-state, i.e., 
\begin{align}
   \mathbb{P} [F_{n,e}=1] & = \lim_{t \rightarrow \infty} \mathbb{P}[F_{n,e}(t)=1], \\ 
    & = \lim_{t \rightarrow \infty} \mathbb{P} [X(t) = \tilde{X}(t)],
    \end{align}
    and similarly,
    \begin{align}
   \mathbb{P} [F_{n,e}=0] & = \lim_{t \rightarrow \infty} \mathbb{P} [X(t) \neq \tilde{X}(t)].
\end{align}
\color{black}
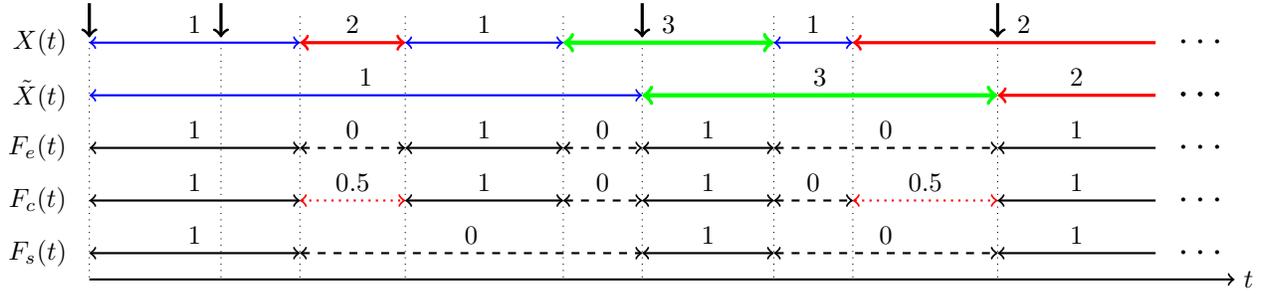
\begin{figure*}[t]
	\centering
	\begin{tikzpicture}[scale=0.35]
		\tikzstyle{note} = [rectangle, dashed, draw, fill=white, font=\footnotesize,
		text width=20em, text centered, rounded corners, minimum height=18em]
        \draw[dotted] (0,17) -- (0,7) ;
        \draw[dotted] (5,17) -- (5,7) ;
        \draw[dotted] (8,17) -- (8,7) ;
        \draw[dotted] (12,17) -- (12,7) ;
        \draw[dotted] (18,17) -- (18,7) ;
        \draw[dotted] (21,17) -- (21,7) ;
        \draw[dotted] (26,17) -- (26,7) ;
        \draw[dotted] (29,17) -- (29,7) ;
        \draw[dotted] (34.5,17) -- (34.5,7) ;
		\draw (-0.5,16) node[anchor=east] {{$X(t)$}} ;
        \draw[thick,blue,<->] (0,16) -- (8,16) ;
        \draw[very thick,red,<->] (8,16) -- (12,16) ;
        \draw[thick,blue,<->] (12,16) -- (18,16) ;
		\draw[ultra thick,green,<->] (18,16) -- (26,16) ;
        \draw[thick,blue,<->] (26,16) -- (29,16) ;
        \draw[very thick,red,<-] (29,16) -- (40.5,16) ;
        \draw (4,16) node[anchor=south] {$1$};	
        \draw (10,16) node[anchor=south] {$2$};	
        \draw (15,16) node[anchor=south] {$1$};	
        \draw (22,16) node[anchor=south] {$3$};	
        \draw (27.5,16) node[anchor=south] {$1$};	
        \draw (35.5,16) node[anchor=south] {$2$};	
        \draw[very thick,->] (0,17.5) -- (0,16.2) ;
        \draw[very thick,->] (5,17.5) -- (5,16.2) ;
        \draw[very thick,->] (21,17.5) -- (21,16.2) ;
        \draw[very thick,->] (34.5,17.5) -- (34.5,16.2) ;
        \draw (41,16) node[anchor=west] {{\Large $\cdots$}};
        \draw (-0.5,14) node[anchor=east] {{$\tilde{X}(t)$}} ;
        \draw[thick,blue,<->] (0,14) -- (21,14) ;
        \draw[ultra thick,green,<->] (21,14) -- (34.5,14) ;
        \draw[very thick,red,<-] (34.5,14) -- (40.5,14) ;
        \draw (10.5,14) node[anchor=south] {$1$};	
        \draw (27.75,14) node[anchor=south] {$3$};	
        \draw (37.5,14) node[anchor=south] {$2$};	
        \draw (41,14) node[anchor=west] {{\Large $\cdots$}};
        \draw (-0.5,12) node[anchor=east] {{$F_e(t)$}} ;
        \draw[thick,black,<->] (0,12) -- (8,12) ;
        \draw[thick,black,dashed,<->] (8,12) -- (12,12) ;
        \draw[thick,black,<->] (12,12) -- (18,12) ;
        \draw[thick,black,dashed,<->] (18,12) -- (21,12) ;
        \draw[thick,black,<->] (21,12) -- (26,12) ;
        \draw[thick,black,dashed,<->] (26,12) -- (34.5,12) ;
        \draw[thick,black,<-] (34.5,12) -- (40.5,12) ;
        \draw (41,14) node[anchor=west] {{\Large $\cdots$}};        
        \draw (4,12) node[anchor=south] {$1$};	
        \draw (10,12) node[anchor=south] {$0$};	
        \draw (15,12) node[anchor=south] {$1$};	
        \draw (19.5,12) node[anchor=south] {$0$};
        \draw (23.5,12) node[anchor=south] {$1$};	
        \draw (30.25,12) node[anchor=south] {$0$};	
        \draw (37.5,12) node[anchor=south] {$1$};	
        \draw (41,12) node[anchor=west] {{\Large $\cdots$}};
        \draw (-0.5,10) node[anchor=east] {{$F_c(t)$}} ;
        \draw[thick,black,<->] (0,10) -- (8,10) ;
        \draw[thick,dotted,red,<->] (8,10) -- (12,10) ;
        \draw[thick,black,<->] (12,10) -- (18,10) ;
        \draw[thick,black,dashed,<->] (18,10) -- (21,10) ;
        \draw[thick,black,<->] (21,10) -- (26,10) ;
        \draw[thick,dashed,<->] (26,10) -- (29,10) ;
        \draw[thick,red,dotted,<->] (29,10) -- (34.5,10) ;
        \draw[thick,black,<-] (34.5,10) -- (40.5,10) ;
        \draw (41,14) node[anchor=west] {{\Large $\cdots$}};        
        \draw (4,10) node[anchor=south] {$1$};	
        \draw (10,10) node[anchor=south] {$0.5$};	
        \draw (15,10) node[anchor=south] {$1$};	
        \draw (19.5,10) node[anchor=south] {$0$};
        \draw (23.5,10) node[anchor=south] {$1$};	
        \draw (27.5,10) node[anchor=south] {$0$};	
        \draw (31.75,10) node[anchor=south] {$0.5$};	
        \draw (37.5,10) node[anchor=south] {$1$};	
        \draw (41,10) node[anchor=west] {{\Large $\cdots$}}; 
        \draw (-0.5,8) node[anchor=east] {{$F_s(t)$}} ;
        \draw[thick,black,<->] (0,8) -- (8,8) ;
        \draw[thick,black,dashed,<->] (8,8) -- (21,8) ;
        \draw[thick,black,<->] (21,8) -- (26,8) ;
        \draw[thick,black,dashed,<->] (26,8) -- (34.5,8) ;
        \draw[thick,black,<-] (34.5,8) -- (40.5,8) ;
        \draw (41,14) node[anchor=west] {{\Large $\cdots$}};
        \draw (4,8) node[anchor=south] {$1$};	
        \draw (14.5,8) node[anchor=south] {$0$};	
        \draw (23.5,8) node[anchor=south] {$1$};	
        \draw (30.25,8) node[anchor=south] {$0$};	
        \draw (37.5,8) node[anchor=south] {$1$};	
        \draw (41,8) node[anchor=west] {{\Large $\cdots$}};
        \draw[thick,->] (0,7) -- (43.5,7)  ;
        \draw (43.5,7) node[anchor=west] {$t$};
	\end{tikzpicture}
	\caption{Sample paths of the processes $X(t)$, $\tilde{X}(t)$, $F_e(t)$, $F_c(t)$ and $F_s(t)$ for a source with states 1, 2 and 3, for an example scenario. Down arrows represent sampling operation.}
	\label{fig:samplepath}
\end{figure*}	
In the FWC model, we assume that there may be a value when the original process and its estimator are close enough to each other despite being out of sync, from a certain semantic perspective. For this purpose, for FWC, we introduce a proximity matrix ${P_n} = \{p_{n,ij}\}$, $0 \leq p_{n,ij} \leq 1$ for source-$n$, and subsequently define a non-binary freshness process $F_{n,c}(t)$ which takes the value $p_{n,ij}$ when $X_n(t)=i$ and $\tilde{X}_n(t) = j$. In particular, $p_{n,ii}=1$, representing perfect freshness when the original process and its estimator are synchronized. Close to unity values of $p_{n,ij}$ are representative of proximity between the states $i$ and $j$. When ${P_n}$ is taken as the identity matrix, FWC reduces to FWE.  For the FWS model, the binary freshness process $F_{n,s}(t)$ is set whenever the process is sampled, and it stays set until $X_n(t)$ makes a transition at which instant $F_{n,s}(t)$ becomes zero.  The mean freshness for FWC (resp. FWS) is denoted by $\mathbb{E}[F_{n,c}]$ (resp. $\mathbb{E}[F_{n,s}])$ similar to \eqref{freshnessDefine} where 
$F_{n,c}$ (resp. $F_{n,s}$) is the random variable associated with the random process $F_{n,c}(t)$ (resp. $F_{n,s}(t)$) in the steady-state.

When the index of the source is immaterial, the subscript $n$ is dropped for the source process $X(t)$ and its estimator $\tilde{X}(t)$  along with the corresponding freshness processes $F_e(t)$, $F_c(t)$, and $F_s(t)$ for FWE, FWC, and FWS, respectively, and the proximity matrix ${P}=\{ p_{ij} \}$ for FWC. 
Fig.~\ref{fig:samplepath} depicts the sample paths of the processes $X(t)$, $\tilde{X}(t)$, $F_e(t)$, $F_c(t)$ and $F_s(t)$ for a source with three states $\{1, 2, 3\}$ for an example scenario where the proximity matrix $P$ is chosen such that $p_{ij}=1$ when $|i-j|=0$, $0.5$ when $|i-j|=1$, and zero when $|i-j|=2$. Note that when $F_s(t)=1$, then $F_e(t)=1$, but not otherwise. Therefore, $\mathbb{E}[F_{e}] \geq  \mathbb{E}[F_{s}]$ for any choice of the sampling rate $\lambda$. Moreover, $\mathbb{E}[F_{c}] \geq  \mathbb{E}[F_{e}]$ since $F_c(t)$ can be larger than zero when $F_e(t)=0$, stemming from the structure of the proximity matrix $P$. Thus, it always holds that $\mathbb{E}[F_{c}] \geq  \mathbb{E}[F_{e}] \geq  \mathbb{E}[F_{s}]$.

\section{Analytical Expressions for Mean Freshness} \label{section:analytical}
The subscript $n$ indicating the source index is dropped for convenience in the current section where the mean freshness is derived for a single irreducible CTMC $X(t)$ for the three freshness models.
\color{black}
\subsection{FWE Model}\label{subsection:FWE}
Theorem~\ref{theorem1} provides an expression for  $f(\lambda) = \mathbb{E}[F_e]$ for the FWE model for the CTMC $X(t)$.
\begin{theorem} \label{theorem1}
Let the irreducible CTMC $X(t) \in \{ 1,2,\ldots,K \}$ with generator $Q$ and steady-state vector $\pi$, be Poisson sampled with sampling rate $\lambda$. Then, for the FWE model, the mean freshness $f(\lambda) = \mathbb{E}[F_e]$ is given by,
	\begin{align}
	f(\lambda)	& = \lambda \pi \ {\bm diag}[(\lambda I - Q)^{-1}], \label{thm1}
	\end{align}
where $\bm{diag}[\cdot]$ represents a column vector composed of the diagonal entries of its matrix argument.
\end{theorem}
\begin{proof} Let us consider the two-dimensional random process $Y(t) = (\tilde{X}(t),X(t))$, which is also Markov. To see this, note that, the transition intensity from state $(i,j)$ to $(i,j')$ is $q_{jj'}$ and from state $(i,j)$ to $(j,j)$ for $j \neq i$ is $\lambda$. Let the steady-state vector of the process $Y(t)$ be denoted by $y$, i.e., $y_{ij}  = \lim_{t \rightarrow \infty} \mathbb{P} [ Y(t) = (i,j)]$, $1 \leq i,j \leq K$. Let $Y$ be a $K \times K$ matrix such that $Y = \{ y_{ij} \}$. Applying the GBE for the state $(i,i)$ of $Y(t)$ provides the following equation for each $i$, $1 \leq i \leq K$,
\begin{align}
 y_{ii} \sigma_i & =  \sum_{j \neq i} y_{ij} q_{ji} +  \sum_{j \neq i} y_{ji} \lambda, \\
  & = \sum_{j \neq i} y_{ij} q_{ji}+ \lambda (\pi_i -y_{ii}),  \label{eqn1}
\end{align}
where the last equality stems from the identity $\pi_i= \sum_{j} y_{ji}$. On the other hand, when the GBE is applied for the state $(i,k)$, $k \neq i$, then we obtain the following,
\begin{align}
y_{ik} (\sigma_k + \lambda) = \sum_{j \neq k} y_{ij} q_{jk}, \quad 1 \leq i \leq K, \ k \neq i.\label{eqn2}
\end{align}
Writing the equations \eqref{eqn1} and \eqref{eqn2} in a matrix form, we obtain for each $i$, $1 \leq i \leq K$,
\begin{align}
	Y(i,:) (Q - \lambda I)  & = -\lambda \pi_i  I(i,:) , \label{eqn3}
\end{align}
where $Y(i,:)$ and $I(i,:)$ denote the $i$th row of $Y$ and the $i$th row of the identity matrix, respectively. In FWE, the freshness process $F_e(t)=1$ when the joint process $Y(t)$ is visiting state $(i,i)$ (in the steady-state) for some state $i$ and $F_e(t)=0$ otherwise. Therefore, 
\(
\mathbb{E}[F_e]  = \sum_{i=1}^K y_{ii},
\)
which yields the identity \eqref{thm1}.
\end{proof}
The following corollary gives a simplified expression for mean freshness $f(\lambda)$ in terms of the sum of first-order rational functions of the variable $\lambda$ for time-reversible CTMCs on the basis of Theorem~\ref{theorem1}.
\begin{corollary} \label{theorem2} 
	Consider the process $X(t)$ of Theorem~\ref{theorem1} with generator $Q$ which is time-reversible and with diagonalizing transformation matrix $T$ as given in \eqref{diagonalization}. Then, the mean freshness $f(\lambda)=\mathbb{E}[F_e]$ is given for the FWE model by 
	\begin{align}
 	f(\lambda) & = \sum_{j=1}^{K} \pi_j^2 + \lambda \sum_{j=1}^{K-1} \frac{b_j}{\lambda + d_j}, \label{thm21} \\
 	& = 1 - \sum_{j=1}^{K-1} \frac{a_j}{\lambda + d_j}, \label{thm22}
	\end{align} 
where $a_j$ and $b_j$, for $1\leq j \leq K-1$, are given by
\begin{align}
b_j & = \sum_{i=1}^{K} {\pi}_i^2 t_{ij}^2 = \sum_{i=1}^{K} {\tilde{t}_{ji}^2}, \quad a_j=  b_j d_j. \label{nail78}
\end{align}
Moreover, $f(\lambda)$ is increasing and strictly concave, and has a continuous derivative $f^{\prime}(\lambda)$ with $\lim_{\lambda \rightarrow 0^+}f(\lambda) = \sum_{j=1}^{K} \pi_j^2$.
\end{corollary}
\begin{proof}
Using the diagonalization equation \eqref{diagonalization}, we first write
the term $A=(\lambda I - Q)^{-1}= T (\lambda I-D)^{-1}T$ appearing in \eqref{thm1} as follows,
\begin{align}
 A &= T
 \mqty(\dmat{\frac{1}{\lambda+d_1},\ddots,\frac{1}{\lambda + d_{K-1}},\frac{1}{\lambda} })  T^{-1}. \label{temp1}
\end{align}
Using \eqref{thm1} and \eqref{temp1}, we have
\begin{align}
\frac{f(\lambda)}{\lambda} & = \sum_{i=1}^{K} {\pi}_i A_{ii}, \\
& =   \sum_{i=1}^{K} {\pi}_i \sum_{j=1}^{K} t_{ij} \tilde{t}_{ji} \frac{1}{\lambda + d_j}, \\
& = \sum_{j=1}^{K-1} \underbrace{\sum_{i=1}^{K} {\pi}_i^2  t_{ij}^2 }_{b_j} \frac{1}{\lambda + d_j} + \frac{1}{\lambda} \sum_{i=1}^{K} \pi_i^2, \\
& = \sum_{j=1}^{K-1} \frac{b_j}{\lambda + d_j} + \frac{1}{\lambda} \sum_{i=1}^{K} \pi_i^2, \label{temp2}
\end{align}
since $\tilde{t}_{ji} = \pi_i t_{ij}$ and also
$\sum_{i=1}^{K} {\pi}_i  t_{i K} \tilde{t}_{K i} = \sum_{i=1}^{K} \pi_i^2$ by observing that $t_{i K} =1$ and $\tilde{t}_{K i} = {\pi}_i$. The result in \eqref{temp2} gives the desired expression in \eqref{thm21}. Then, \eqref{thm22} follows directly from \eqref{temp2} and also the fact that $\lim_{\lambda \rightarrow \infty} f(\lambda) =1$.  
Moreover, the coefficients $a_j$ and $b_j$ are strictly positive since $\pi_i>0$ and the entries of a column of $T$ cannot be all zero. A first-order rational function of $\lambda$ in the form $-a/(\lambda + d)$ is increasing and strictly concave for $a,d>0$ and sums of concave functions are also increasing and strictly concave, completing the proof. The expression  pertaining to $\lim_{\lambda \rightarrow 0^+}f(\lambda)$ immediately follows from \eqref{thm21}.
\end{proof}
For the special case $K=2$, $X(t)$ is a two-state time-reversible CTMC with generator $Q$ and the diagonal matrix $\Pi$ given as follows \eqref{Pi}, 
\begin{align}
Q=
\begin{pmatrix}
 -\alpha & \alpha \\
 \beta & -\beta
\end{pmatrix},
\quad
\Pi=
\begin{pmatrix}
 \frac{\beta}{\alpha + \beta} & 0 \\
 0 & \frac{\alpha}{\alpha + \beta}
\end{pmatrix}.
\end{align}
The generator $Q$ has two eigenvalues: the first eigenvalue being $-d_1$ where $d_1 = \alpha +\beta$ and the second one at the origin. Consequently, we obtain the symmetric matrix $S$ according to \eqref{S} and the orthogonal transformation matrix $U$ from \eqref{diagonalization0},
\begin{align}
S=
\begin{pmatrix}
 -\beta & \sqrt{\alpha\beta} \\
 \sqrt{\alpha\beta}  & -\beta
 \end{pmatrix}, \quad
 U = 
 \begin{pmatrix}
 \sqrt{\frac{\alpha}{\alpha+\beta}} & \sqrt{\frac{\beta}{\alpha+\beta}} \\
 \sqrt{\frac{\beta}{\alpha+\beta}} & \sqrt{\frac{\alpha}{\alpha+\beta}}
\end{pmatrix},
\end{align}
which gives rise to the matrix $\tilde{T}$ which is the inverse of the diagonalizing transformation matrix $T$ (see \eqref{diagonalization}),
\begin{align}
\tilde{T}=
 \begin{pmatrix}
 {\frac{\alpha\beta}{(\alpha+\beta)^2}} & {\frac{\alpha\beta}{(\alpha+\beta)^2}} \\
 {\frac{\beta}{\alpha+\beta}} & {\frac{\alpha}{\alpha+\beta}}
\end{pmatrix}.
\end{align}
From \eqref{thm22}, the following closed-form solution exists for mean freshness,
\begin{align}
f(\lambda) 	& = 1 - \frac{a_1}{\lambda+d_1},
\label{staleness}
\end{align}
where $a_1 =\frac{2 \alpha \beta}{(\alpha + \beta)}$ since $a_1 = b_1 d_1$ and $b_1$ is the sum of the squares of the entries of the first row of $\tilde{T}$ which is evident from the equations \eqref{thm22} and \eqref{nail78}. We note that expressions for expected staleness for this limited special case of 2-state sources have been obtained in \cite{bastopcu_ulukus_entropy22} using a different method.

\subsection{FWC Model} \label{subsection:FWC}
Corollary~\ref{theorem3} provides an expression for $f(\lambda)=\mathbb{E}[F_c]$ for the FWC model for time-reversible CTMCs on the basis of Theorem~\ref{theorem1}.
\begin{corollary} \label{theorem3} 
	Consider the process $X(t)$ of Theorem~\ref{theorem1} with generator $Q$ which is time-reversible and with diagonalizing transformation matrix $T$ as given in \eqref{diagonalization}. Then, the mean freshness $f(\lambda)=\mathbb{E}[F_c]$ with proximity matrix $P$ is given for the FWC model by 
	\begin{align}
 	f(\lambda) & = \pi \sum_{i=1}^{K} \pi_i P(:,i) +  \sum_{j=1}^{K-1} \frac{b_j \lambda}{\lambda + d_j}, \label{thm31} \\
 	& = 1 - \sum_{j=1}^{K-1} \frac{a_j}{\lambda + d_j}, \label{thm32}
	\end{align} 
    where $a_j$ and $b_j$, for $1\leq j \leq K-1$, are given by
    \begin{align}
    b_j & = \tilde{T}(j,:) \sum_{i=1}^{K} {\pi}_i t_{ij}  P(:,i), \quad a_j=  b_j d_j, 
    \end{align}
    and $P(:,i)$ denotes the $i$th column of $P$, and $\lim_{\lambda \rightarrow 0^+}f(\lambda) =\pi \sum_{i=1}^{K} \pi_i P(:,i)$.
\end{corollary}

\begin{proof}
Recalling the definition of $y_{ij}$, we write 
\begin{align}
f(\lambda) & = \sum_{i=1}^K  \sum_{j=1}^K   y_{ij} p_{ji}.
\end{align} 
Recalling the definition of matrix $A$, we first obtain the following identity from \eqref{eqn3},
\begin{align}
\frac{f(\lambda)}{\lambda} & = \sum_{i=1}^{K} {\pi}_i A(i,:) P(:,i).
\end{align}
Consequently, 
\begin{align}
 \frac{f(\lambda)}{\lambda}& =   \sum_{i=1}^{K} {\pi}_i \left( \sum_{j=1}^{K} t_{ij} \tilde{T}(j,:) \frac{1}{\lambda + d_j} \right)  P(:,i), 
\end{align}
which is equal to the following expression,
\begin{align} & \sum_{j=1}^{K-1} \underbrace{ \tilde{T}(j,:) \sum_{i=1}^{K}  {\pi}_i t_{ij}  P(:,i)  }_{b_j} \frac{1}{\lambda + d_j} + \pi \sum_{i=1}^{K} 
 \frac{\pi_i P(:,i)}{\lambda}, \label{temp8}
\end{align}
since $t_{i K} =1$ and $\tilde{t}_{K i} = {\pi}_i$, giving the desired expression in \eqref{thm31}. Then, \eqref{thm32} follows directly from \eqref{temp8} and also from $\lim_{\lambda \rightarrow \infty} f(\lambda) =1$. However, the coefficients $b_j$ (and hence $a_j$) are not guaranteed to be non-negative and some of these coefficients may indeed be negative.
\end{proof}

\begin{remark}
Although FWE is a sub-case of FWC, we presented the results for FWE separately since in this case the expressions are slightly simpler and the coefficients $a_j$'s are shown to be non-negative ensuring concavity of the expression \eqref{thm22}.
\end{remark}

\subsection{FWS Model} \label{subsection:FWS}
Theorem~\ref{theorem4} provides an expression for $f(\lambda)=\mathbb{E}[F_s]$ for the FWS model for the CTMC $X(t)$.
\begin{theorem} \label{theorem4}
    Let the irreducible CTMC $X(t) \in \{ 1,2,\ldots,K \}$ with generator $Q$ and steady-state vector $\pi$, be Poisson sampled with sampling rate $\lambda$. Then, for the FWS model, the mean freshness $f(\lambda)=\mathbb{E}[F_s]$ is given by,
	\begin{align}
		f(\lambda) & = 1 - \sum_{i=1}^{K} \frac{a_i}{\lambda + \sigma_i}, \label{thm4}
	\end{align} 
    where $a_i = \pi_i \sigma_i$. Moreover, $f(\lambda)$ is increasing and strictly concave, and has a continuous derivative $f^{\prime}(\lambda)$ with $\lim_{\lambda \rightarrow 0^+}f(\lambda) =0$.
\end{theorem}
\begin{proof}
Consider the two-dimensional process $Z(t) = (F_s(t),X(t))$, which is Markov. To see this, in FWS, the transition intensity from states $(1,j)$ and $(0,j)$ to the states $(1,j')$ and $(0,j')$, respectively, is $q_{j j'}$. On the other hand, the transition intensity from  state $(0,j)$ to $(1,j)$ is $\lambda$.
Let the steady-state solution of the process $Z(t)$ be denoted by $z$, i.e., $z_{ij}  = \lim_{t \rightarrow \infty} \mathbb{P}[ F_s(t) = i,X(t)=j],
0 \leq i \leq 1, 1 \leq j \leq K$. We show that the following choice of $z$
\begin{align}
	z_{0i} & =  \frac{\pi_i \sigma_i}{\lambda + \sigma_i}, \quad
	z_{1i} =  \frac{\pi_i \lambda}{\lambda + \sigma_i}, \label{temp21}
\end{align} 
satisfies the following GBE for the states $(1,i)$, $1 \leq i \leq K$,
\begin{align*}
z_{1i} \sigma_i & = z_{0i} \lambda,	
\end{align*}
which is a direct result of \eqref{temp21}.
In order to show that $z$ defined as in \eqref{temp21} satisfies the GBE for the states
$(0,i)$, $1 \leq i \leq K$, we write from \eqref{temp21},
\begin{align}
	z_{0i} (\sigma_i + \lambda) & = 	\pi_i \sigma_i = \sum_{j \neq i} \pi_j q_{ji}, \\
	& = \sum_{j \neq i} \left( \frac{\lambda}{\lambda + \sigma_j} + \frac{\sigma_j}{\lambda + \sigma_j} \right) \pi_j  q_{ji}, \\
	& = \sum_{j \neq i} (z_{0j} + 	z_{1j}) q_{ji}.
\end{align}	
Moreover, $\sum_{i=0}^1 \sum_{j=1}^{K} z_{ij} =1 $, and therefore, $z$ as given in \eqref{temp1} is the steady-state solution for the CTMC $Z(t)$. The mean freshness is finally expressed as, 
\(
	f(\lambda)  = 1 -  \sum_{j=1}^{K} z_{0j},
\)
which yields \eqref{thm4}. Since the form of expression is the same as in the FWE model for time-reversible CTMCs, $f(\lambda)$ is an increasing and strictly concave function of $\lambda$. Moreover, $\lim_{\lambda \rightarrow 0^+}f(\lambda) = 1-\sum_i \pi_i=0$ from \eqref{thm4}.
\end{proof}	

\section{Optimum Monitoring of Heterogeneous CTMCs} \label{section:Optimization}
The monitor is resource-constrained, and therefore, there is a constraint $\Lambda$ on the overall sampling rate of the monitor, i.e., $\lambda=\sum_{n=1}^N \lambda_n \leq \Lambda$. 
Let us first focus our attention to the FWE freshness model for time-reversible CTMCs in which case we use the mean freshness metric $f_n(\lambda_n) = \mathbb{E}[F_{n,e}]$ for source-$n$, and the weighted sum freshness (or the system freshness) $F_S = \sum_{n=1}^{N} w_n f_n(\lambda_n)$, for the overall monitoring system where the normalized weights $w_n, n=1,\ldots,N, \; \sum_n w_n =1,$ reflect the relative importance of the freshness of the information processes. Thus, we have the following optimization problem for weighted sum freshness maximization,
\begin{maxi}
	{\lambda_n \geq 0}{\sum_{n=1}^{N} w_n f_n(\lambda_n) = 1 - \sum_{n=1}^{N} \sum_{j=1}^{K_n - 1} \frac{w_n a_{n,j}}{\lambda_n + d_{n,j}}} 
	{\label{Optimization1}}
    {}
	\addConstraint{ \sum_{n=1}^{N} \lambda_n}{\leq \Lambda}
\end{maxi}
In \eqref{Optimization1}, the coefficients $a_{n,j},d_{n,j} >0$, for $1 \leq j \leq K_n - 1$ are to be obtained for the CTMC $X_n(t)$ using the procedure described in Corollary~\ref{theorem2} and the expression \eqref{thm22}. The function $f_n(\lambda_n)$ is increasing and strictly concave, and has a continuous first order derivative $f_n^{\prime}(\lambda_n)$ that monotonically decreases from the value $\infty$ at $\lambda_n=-d_n^ {\ast}$ to zero as $\lambda_n$ is increased to $\infty$, where $d_n^ {\ast} = \min_j d_{n,j}$. This optimization problem is known to have a water-filling solution \cite{WF_Poor} on the basis of which Algorithm~\ref{alg:Water-filling} provides an efficient solution to the optimization problem \eqref{Optimization1} which requires at most $N-1$ iterations until termination. Step~2 of Algorithm~\ref{alg:Water-filling} can be solved by using the two-dimensional bisection search algorithm detailed in \cite{WF_Poor}. 

The algorithm is outlined as follows. Initially, $I_n =1$ for $n=1,\ldots,N$. Then, for a given $\mu > 0$ and for each $n$ such that $I_n=1$, we iteratively find the value of $\lambda_n \in (-d_n^ {\ast},\infty)$ that satisfies $w_n f_n^{\prime} (\lambda_n) = \mu$ using an inner bisection search algorithm. Once $\lambda_n$'s are obtained, we check whether $\sum_{n=1}^{N}\lambda_n I_n < \Lambda$ or not, and we vary the value of $\mu$ according to an outer bisection search algorithm, which iteratively finds the value of $\mu$ such that  $\sum_{n=1}^{N}\lambda_n I_n = \Lambda$. If $\lambda_n \leq 0$ at this step, then $I_n$ and $\lambda_n$ are set to zero for all such $n$ and the procedure above is repeated.

For the special case of two-state CTMCs, i.e., $K_n = 2$, a closed-form solution is available for the solution of the equations in Step~2 since the inverse function of $f_n^{\prime}(\cdot)$ can be written in closed form. In this case, it is not difficult to show that the choices of $\mu$ and $\lambda_n$ for sources with $I_n=1$, 
\begin{align} 
	\mu = \left(\frac{\sum_{n=1}^N \sqrt{w_{n,1} a_{n,1} I_n}}{\Lambda + \sum_{n=1}^N d_{n,1} I_n} \right)^2, \;  \lambda_n  = \sqrt{\frac{w_n a_{n,1}}{\mu}} - d_{n,1},  \label{nail2}
\end{align}  
provide a single-shot solution for Step~2 of Algorithm~\ref{alg:Water-filling} without a requirement for bisection search for this step. 

We note that, for the FWS model, Algorithm~\ref{alg:Water-filling} can be used with the only difference being the upper limit of the inner summation changed to $K_n$ in \eqref{Optimization1}. 
For the FWC model, since the coefficients $a_i$'s in \eqref{thm32} can be negative, concavity of the freshness function is not proven. However, we propose to use the same water-filling algorithm also for the FWC model based on the observation that the expression \eqref{thm32} turned out to be concave in all the examples we studied.

\begin{algorithm}[bt]
	\caption{Water-filling algorithm for the optimization problem \eqref{Optimization1} based on \cite{WF_Poor}.}
	\begin{algorithmic}
		\STATE \textbf{Step 1}: Initialize $I_n =1$ for $n=1,\ldots,N$.
		\STATE \textbf{Step 2}: Solve the following equations for the water level $\mu > 0$ and $\lambda_n$ when $I_n=1$,
		\begin{align}
			 w_n f_n^{\prime} (\lambda_n) = \mu, \;  \sum_{n=1}^{N}\lambda_n I_n = \Lambda.  \label{eqn:wf}
		\end{align}
		\STATE \textbf{Step 3}: If $\lambda_n > 0$ for all $n$ such that $I_n=1$, then terminate while returning $\lambda_n$'s. 
		\STATE \textbf{Step 4}: Otherwise, set $I_n=0$ and $\lambda_n=0$ for all $n$ such that $I_n=1$ and $\lambda_n \leq 0$, and go to Step~2.
	\end{algorithmic}
	\label{alg:Water-filling}
\end{algorithm}

\section{Numerical Examples} \label{section:Examples}
The first numerical example is presented to validate the analytical expressions obtained for mean freshness in corollaries~\ref{theorem2} and \ref{theorem3}, and Theorem~\ref{theorem4}. A time-reversible birth-death CTMC with three states is considered with generator
\begin{align}
Q= \begin{pmatrix}
    -1.95  &  1.95   &    0 \\
    1  & -2.95 &  1.95 \\
         0  &  2  & -2 
\end{pmatrix}.
\end{align}
For the FWC model, we use the proximity matrix of Fig.~\ref{fig:samplepath}.
The mean freshness is first obtained in Fig.~\ref{Validation} as  a function of $\lambda$ for the three freshness models FWE, FWC, and FWS, using the analytical expressions in \eqref{thm22}, \eqref{thm32}, and \eqref{thm4}, respectively, which is termed as the analytical (A) method. Note that in the analytical method, the expressions are first obtained once for the CTMC which are then used to obtain the metrics for any given sampling rate $\lambda$. On the other hand, the same freshness metrics can also be obtained by numerically solving the two-dimensional Markov chains (used in the proofs of theorems~\ref{theorem1} and \ref{theorem4}) constructed for each given $\lambda$. The results match perfectly, validating the analytical method. 

In the second numerical example, we focus on FWE and FWS in a scenario of $N=50$ heterogeneous two-state Markov chains with $\pi_{n,1}=0.3$, $\pi_{n,2}=0.7$ and linearly spaced transition intensities, i.e., $r_n = r_{n-1} + \delta$, $n=2,\ldots,N$. In the numerical example, we set $r_{1}=0.01$ and the average transition intensity $\frac{1}{N} r = \frac{1}{N} \sum_{n=1}^N r_n$ is set to 10 which yields the choice of $\delta=0.4078$.  We denote by $\kappa$ the ratio of the overall sampling rate $\Lambda$ to the system transition intensity, i.e., $\kappa = \frac{\Lambda}{r}$. Obviously, the sampling ratio $\kappa$ should be sufficiently large so as to keep the remote estimates of all the information sources fresh. The source weights are assumed to be the same with $w_n=\frac{1}{N}$. Algorithm~\ref{alg:Water-filling} is used for both FWE and FWS models, whereas for FWE, \eqref{nail2} is employed for Step~2 of the algorithm to obtain the optimum sampling rates $\lambda_n$'s under an overall sampling rate constraint $\Lambda$ which is chosen to attain a given sampling ratio $\kappa$. 
\begin{figure}[tb]
	\centering
	\includegraphics[width=0.75\columnwidth]{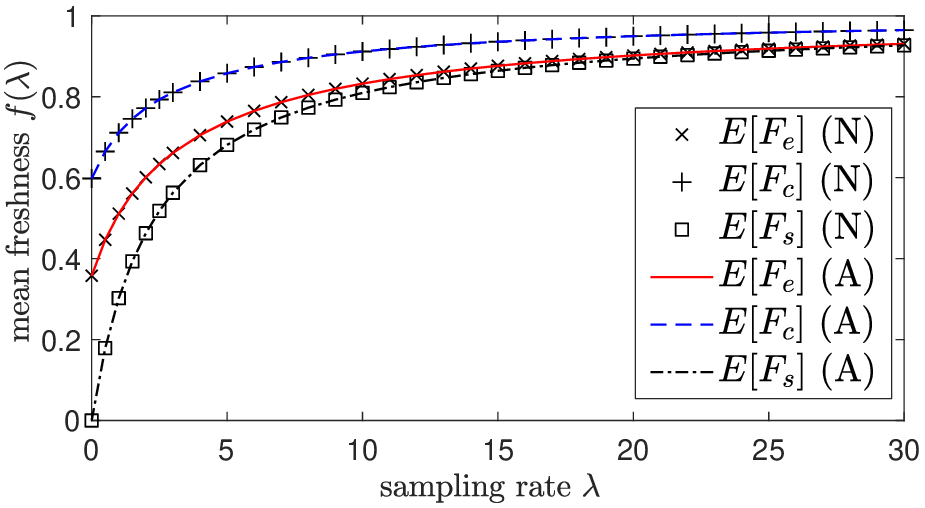}
	\caption{The mean freshness as a function of the sampling rate $\lambda$ for the  FWE, FWC, and FWS freshness models using the numerical (N) and analytical (A) methods.}
	\label{Validation}
\end{figure}
We compare our proposed water-willing solution, namely WF, with three baseline policies: i) UNIFORM policy samples each source-$n$ uniformly likely, i.e., $\lambda_n = \frac{\Lambda}{N}$, ii) PROP policy chooses the sampling rate $\lambda_n$ proportional with the source's transition intensity $r_n$, i.e., $\lambda_n \propto r_n$, iii)  INVPROP policy chooses the sampling rate $\lambda_n$ inversely proportional with the source's transition intensity $r_n$, i.e., $\lambda_n \propto \frac{1}{r_n}$. Fig.~\ref{fig1} depicts the system freshness $F_S$ as a function of the sampling ratio $\kappa$ when the WF, UNIF, PROP, and INVPROP sampling policies are employed for both FWE and FWS freshness models. For both models, we observe that the WF policy outperforms all the other three baseline policies. The UNIFORM policy yields very close to optimum freshness performance when the sampling ratio increases. However, for low sampling ratios, it is substantially outperformed by the WF policy. The PROP and INVPROP sampling policies perform poorly for small and large sampling ratios, respectively, against all other policies.

\begin{figure*}[tbh]
	\centering
	\includegraphics[width=14cm]{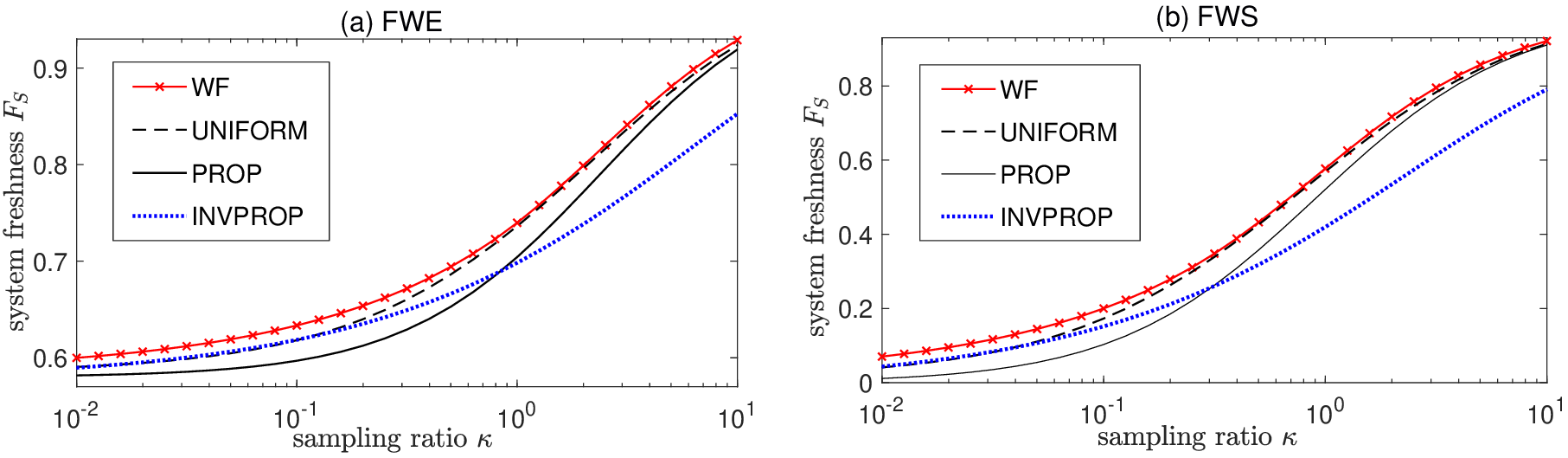}
	\caption{The system freshness $F_S$ as a function of the sampling ratio $\kappa$ for the WF, UNIFORM, PROP, and INVPROP sampling policies: (a) FWE (b) FWS.}
	\label{fig1}
\end{figure*}

Fig.~\ref{fig2} depicts the optimum sampling rate $\lambda_n$ divided by the sampling ratio $\kappa$ as a function of the source index $n$ for different values of the sampling ratio $\kappa$ for both freshness models. We observe that, for low sampling ratios, the water-filling solution chooses not to sample at all, a portion of the sources with high transition intensities, for both models. However, when the sampling ratio is sufficiently high, e.g., $\kappa=10$, the optimum sampling rate for a given source appears to be monotonically increasing with the transition intensity of the source. Although the general behavior of the optimum sampling rate with respect to source index is quite similar for FWE and FWS models, in the latter model, the optimum sampling rate is more uniform across the sources for FWS than FWE. 

\begin{figure*}[tbh]
	\centering
	\includegraphics[width=14cm]{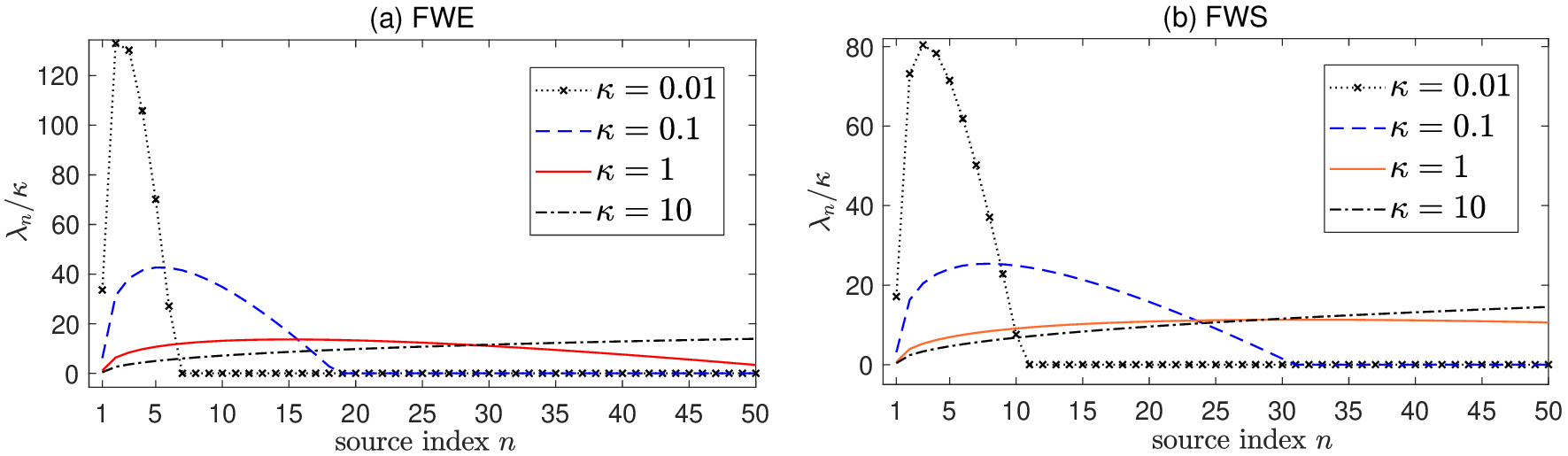}
	\caption{The optimum sampling rate $\lambda_n$ divided by $\kappa$ as a function of the source index $n$ for four values of the sampling ratio $\kappa$: (a) FWE (b) FWS.}
	\label{fig2}
\end{figure*}

\begin{figure*}[tbh]
	\centering
	\includegraphics[width=14cm]{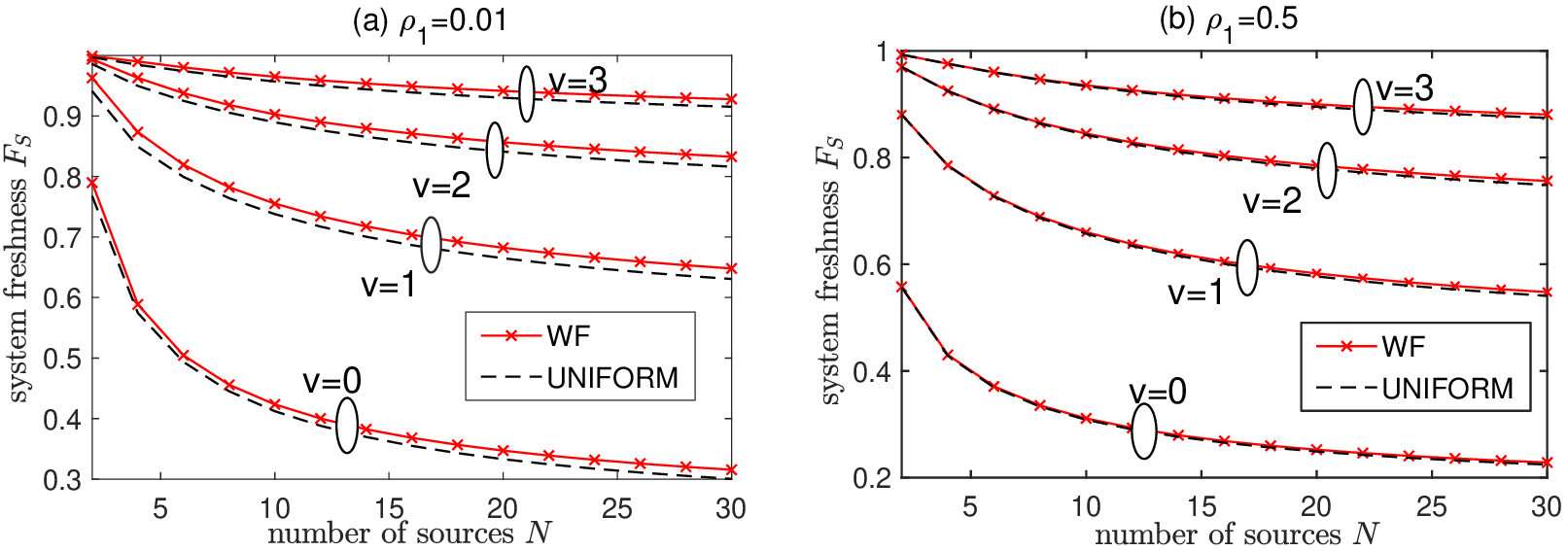}
	\caption{System freshness $F_S$ for the FWC model as a function of the number of users, $N$, for the WF and UNIFORM sampling policies for four different choices of the proximity matrix, i.e., $v \in \{0,1,2,3\}$.} 
	\label{fig3}
\end{figure*}

In the final example, we focus on the FWC model and the choice of the proximity matrix $P$ in terms of a proximity parameter $v$ such that $p_{ij}=1$ when $|i-j| \leq v$, and is zero otherwise. It is clear that FWC with $v=0$ reduces to FWE. We consider an independent collection of 
$N$ CTMCs each of which corresponds to the number of active servers in a multi-server M/M/c/c queuing system with $c$ servers, with common service rate $\gamma$, and arrival rate $\xi_n$ for source-$n$. The load for source-$n$ is denoted by $\rho_n = \xi_n/c \gamma $. In this example, we assume linearly spaced loads, $\rho_n = \rho_{n-1} + \delta$, $n=2,\ldots,N$ and the parameter $\delta$ is chosen so that the average load is fixed to $\rho_{avg} = \frac{1}{N} \sum_{n=1}^N \rho_n$. In this example, we fix $\gamma =1$ and $\rho_{avg}=0.9$. The source weights are identical as in the previous examples. The overall sampling rate bound is fixed to $\Lambda=20$. The weighted sum freshness $F_S= \frac{1}{N} \sum_{n=1}^N \mathbb{E}[F_{n,c}(t)]$ is plotted in Fig.~\ref{fig3} as a function of the number of users $N$ for FWC with the proximity parameter $v \in \{0,1,2,3\}$ and for two values of $\rho_1$ for the WF and UNIFORM policies. We have the following observations: When $\rho_1$ is close to $\rho_{avg}$, then all the sources have similar statistical behaviors and therefore the performances of WF and UNIFORM policies should be similar which is evident from  Fig.~\ref{fig3}(b). However, WF substantially outperforms the UNIFORM sampling policy in  Fig.~\ref{fig3}(a) where the smallest load source-1 has a load $\rho_1=0.01$ and consequently the sources are statistically dissimilar from each other. We have observed similar outperformance behavior of WF over the UNIFORM policy for the four values of the proximity parameter $v$ we have investigated. The weighted sum freshness decreases with increased $N$ since the sampling rate parameter $\Lambda$ is fixed for any choice of $N$. Therefore, sources are sampled at a lower intensity, on the average, as $N$ is increased in our example.

\section{Conclusions} \label{section:Conclusions}
We investigated a remote monitoring system which samples a heterogeneous collection of finite-state irreducible CTMCs according to a Poisson process, under an overall sampling rate constraint, employing a remote martingale estimate of the states of each of the CTMCs. Three binary freshness models are studied and expressions for mean freshness are obtained for all the freshness models of interest. Subsequently, the optimum sampling rates for all CTMCs are obtained using  water-filling based optimization while maximizing the weighted sum freshness. The worst case computational complexity of the proposed method is quadratic in the number of CTMCs making it possible to solve for scenarios even with very large numbers of CTMCs. The optimum monitoring policy is shown to outperform a number of heuristic baseline policies especially when there is diversity in the statistical characteristics of the underlying sources. 
Future work will consist of the study of estimators other than the martingale estimator, information sources other than CTMCs, and the case of partially known source dynamics. Study of optimization techniques other than water-filling, especially for more general information sources (not necessarily time-reversible), is an interesting research direction.
Another possibility is to take into consideration the most recently taken sample values while making a decision on which source to sample.


\end{document}